\documentclass[letter,12pt]{amsart}
\usepackage{amsmath, amsthm, amstext}
\usepackage{amssymb, array, amsfonts}
\usepackage[english]{babel}
\usepackage{bbding}
\usepackage{hyperref}
\usepackage{float}
\usepackage{graphics}
\usepackage{graphicx}
   \usepackage{enumerate}
\numberwithin{equation}{section}

\def \er{\varepsilon}
\def \sconv{\xrightarrow[]{s}}

\renewcommand{\l}{\left}
\renewcommand{\r}{\right}

\def \GL{\mathrm{GL}}

\def \C{\mathbb{C}}
\def \Q{\mathbb{Q}}

\def \N{\mathbb{N}}
\def \M2{\mathrm{M}_2}
\def \R{\mathbb{R}}
\def \Z{\mathbb{Z}}

\def \T{\mathbb{T}}

\def \H{\mathcal{H}}
\def \sl2r{\mathrm{SL}(2,\R)}

\def \dirint {\int_{\T}^{\oplus}}
\def \dirintd {\int_{\T^d}^{\oplus}}
\DeclareMathOperator*{\slim}{s--lim}

\newcommand{\beq}{\begin{equation}}
\newcommand{\eeq}{\end{equation}}

\def\dom{\operatorname{Dom}}
\def\eff{\operatorname{eff}}

\def\wt{\widetilde}

\def\wt{\widetilde}

\DeclareMathOperator*{\esssup}{ess\,sup}

\def\im{\operatorname{Im}}

\newcommand{\eqdef}{\stackrel{\rm def}{=\kern-3.6pt=}}

\theoremstyle{plain}
\newtheorem{theorem}{\bf Theorem}[section]
\newtheorem{lemma}[theorem]{\bf Lemma}

\newtheorem{prop}[theorem]{\bf Proposition}
\newtheorem{cor}[theorem]{\bf Corollary}

\theoremstyle{definition}

\theoremstyle{remark}
\newtheorem{remark}[theorem]{\bf Remark}

\theoremstyle{cond}

\renewcommand{\le}{\leqslant}
\renewcommand{\ge}{\geqslant}

\newcommand{\dist}{\mathop{\mathrm{dist}}\nolimits}

\renewcommand{\qed}{\vrule height7pt width5pt depth0pt}
\title{On transport properties of isotropic quasiperiodic $XY$ spin chains}
\author{Ilya Kachkovskiy}
\begin{document}
	\maketitle
	\begin{abstract}
		We consider isotropic $XY$ spin chains whose magnetic potentials are quasiperiodic
		and the effective one-particle Hamiltonians have absolutely continuous spectra. 
		For a wide class of such $XY$ spin chains, we obtain lower bounds on their 
		Lieb--Robinson velocities in terms of group velocities of their effective Hamiltonians:
		$$
		\mathfrak v\ge\esssup\limits_{[0,1]} \frac{2}{\pi}\frac{dE}{dN},
		$$
		where $E$ is considered as a function of the integrated density of states.
	\end{abstract}
\section{Introduction}
An $XY$ spin chain is one of the most well understood models in many-body quantum 
physics. We will only consider isotropic $XY$ spin chains. For an integer interval $\Lambda=[m,n]\subset \Z$, define the Hamiltonian as
$$
H_{\Lambda}=-\sum\limits_{j=m}^n \l(\sigma_j^x\sigma_{j+1}^{x}+\sigma_j^y\sigma_{j+1}^{y}\r)-\sum\limits_{j=m}^n \nu_j \sigma_j^z,
$$
It acts in the state space
$$
\mathfrak G_{\Lambda}:=\otimes_{l=m}^n \C^2=:\otimes_{j=m}^n \mathfrak G_j,
$$
where $\mathfrak G_j$ is the {\it single site state space} identified with $\C^2$, and the matrices $\sigma_j^{x,y,z}$ are the standard Pauli matrices
$$
\sigma^x=\begin{pmatrix}0&1\\1&0\end{pmatrix},\quad
\sigma^y=\begin{pmatrix}0&-i\\i&0\end{pmatrix},\quad 
\sigma^z=\begin{pmatrix}1&0\\0&-1\end{pmatrix}
$$
acting at the respective sites, so that
$$
\sigma_j^{x,y,z}=I^{\otimes (j-1)}\otimes \sigma^{x,y,z}\otimes I^{\otimes (n-m-j)}.
$$ 
Finally $\{\nu_j\}_{j\in \Lambda}$ is a sequence of real numbers which is the magnetic 
potential.

There have been several interesting developments regarding transport properties in this 
model. The most well known and simple one is the {\it Lieb-Robinson bound}.  If $S\subset \Lambda$, then a {\it local observable} with respect to $S$ is any operator of the form
$$
A\otimes (\otimes_{j\in \Lambda\setminus S}I),
$$
where $A$ acts in $\otimes_{j\in S}\mathfrak G_j$. We denote the algebra of all local observables on $S$ by $\mathcal O(S)$. Note that, formally speaking, this algebra depends on $\Lambda$, but there is a natural correspondence between $\mathcal O(S)$ for different $\Lambda$, so we use the same notation for them. If $A$ is an observable, then
\beq
\label{heisenberg}
A(t):=e^{it H_{\Lambda}}A e^{-it H_{\Lambda}}
\eeq
is the Heisenberg evolution of $A$. Again, it depends on $\Lambda$. However, there are 
certain results (such as the following proposition) that hold for all $\Lambda$, in 
which case we drop the dependence on $\Lambda$ from the notation. The following is 
established in \cite{LiebRob,NachtSims}.
\begin{prop}
	\label{liebrob}
	Suppose that $\{\nu_j\}_{j\in \Z}$ is a bounded real sequence. There exist constants 
	$\eta, \mathfrak{v},C$ depending only on $\sup_{j\in \Z} |\nu_j|$, such that for any finite $\Lambda\subset \Z$ and any two observables $A\in \mathcal O(S_1)$, $B\in \mathcal O(S_2)$, $S_1,S_2\subset \Lambda$, we have
	\beq
	\label{liebrobbound}
	\|[A(t),B]\|\le C\|A\|\|B\| e^{-\eta(\mathfrak{v}t-\dist(S_1,S_2))}.
	\eeq
\end{prop}
The existence of such bound indicates that, even though the propagation speed of the 
Heisenberg evolution is infinite (local observables may become non-local immediately), 
one can still obtain an exponential bound on the tail that is sufficiently far away. 
In other words, the ``physically effective'' speed 
of propagation is still finite, regardless of the potential. Since there is always an 
upper bound, there are three interesting regimes of the transport behavior: the possible 
velocity can be bounded from below, can be made arbitrarily small, or can be made zero. The first case corresponds to the ballistic transport, the last case is related to localization, and the third case is an intermediate situation called anomalous transport.
To describe these properties in more detail, let us note that, in a certain sense, the $XY$ spin chain is a
completely integrable system. The Jordan--Wigner transform reduces the study of this model to the study of  the following {\it effective one-particle Hamiltonian} in $l^2(\Z)$,
$$
(H_{\eff}\psi)_n=\psi_{n+1}+\psi_{n-1}+\nu_j \psi_j.
$$
We refer the reader to \cite{HSS,D_per} for the description of this transformation. We also give some additional remarks in Section 6. 

The transport properties of the spin chain are related to those of the 
effective Hamiltonian. The zero-velocity bound is a consequence of {\it dynamical Anderson localization}, which corresponds to purely point spectrum of the effective Hamiltonian, see, for example, \cite{HSS,SS,CS}.
In the paper \cite{D_fib}, a system with effective quasiperiodic Fibonacci Hamiltonian was extensively studied, and it was established that it displays anomalous transport, where one 
needs to replace $t$ by $t^{\alpha}$ in the exponent of the bound. This corresponds to singular continuous spectrum of $H_{\eff}$. Finally, in \cite{D_per}, the case of periodic 
effective Hamiltonian was studied, and it was established that these systems admit lower 
bound on Lieb--Robinson velocity. This situation corresponds to absolutely continuous spectrum. The result of \cite{D_per} is also proved for anisotropic $XY$ spin chains.

While the transport properties of Schr\"odinger operators on $\Z$ are studied quite well, 
not all of them translate easily to the $XY$ chain case. The reason is that the 
Jordan--Wigner transformation is not local, and so, the bounds on time-averaged transport 
exponents are not sufficient due to possible spreading of wave packets. The result of \cite{D_per} for periodic potentials was obtained by showing that the lower bound on Lieb-Robinson velocity follows from existence of the following strong limit:
\beq
\label{q_def}
Q=\slim_{T\to+\infty}\frac{1}{T}\int_0^T e^{iH_{\eff}t} A e^{-iH_{\eff}t}\,dt,
\eeq
where $(A\psi)_n=i(\psi_{n+1}-\psi_{n-1})$. In this case, in any Lieb-Robinson bound we must have $\mathfrak v\ge 2\|Q\|$\footnote{The estimate in \cite{D_per} has the form $\mathfrak v\ge \|Q\|$ because the effective Hamiltonian in their notation is $2H_{\eff}$ in ours. It is convenient for us to have the off-diagonal part of $H_{\eff}$ being the usual discrete Laplacian.}.
Existence of this operators is the strongest form of ballistic transport: 
it implies that all non-averaged lower transport exponents are equal to 1. The relation 
between $Q$ and transport properties was first observed in \cite{Knauf}.

In the present paper, we study the $XY$ spin chain with quasiperiodic multi-frequency effective Hamiltonian
$$
(H(x)\psi)_n=\psi_{n+1}+\psi_{n-1}+v(x+n\alpha)\psi_n.
$$
where $\alpha$ is an irrational frequency vector and $v$ is a continuous function on a $d$-dimensional torus. We show that, under the assumption of $L^2$ degree 0 reducibility of the corresponding Schr\"odinger cocycle, there exists a non-trivial bound on possible values of $\mathfrak v$ in Theorem \ref{liebrob}. The assumptions of the theorem hold in many cases 
where absolute continuity of the spectrum of $H_{\eff}$ is known: for example, for 
analytic one-frequency potentials with Diophantine frequencies, and for analytic 
multi-frequency potentials at (perturbatively) small coupling. The results on concrete operators are summarized in Corollaries \ref{concrete1} and \ref{concrete2}.

Unlike the case of \cite{D_per}, we initially only establish existence of the phase-averaged 
version of $Q$. This implies that the limit \eqref{q_def} exists on a subsequence of time 
scales, which is still enough to obtain a velocity bound. The drawback is that it does not imply ballistic transport for $H(x)$, but implies one replaced by a phase-averaged version, see Remark \ref{ballexp_rem}.

We also give an explicit description of the lower bound in terms of the group velocity for the effective Hamiltonian:
$$
\mathfrak v\ge\esssup\limits_{[0,1]} \frac{2}{\pi}\frac{dE}{dN},
$$
where $\frac{dE}{dN}$ is the derivative of the inverse function of the integrated density of states of $H_{\eff}$. See Theorem \ref{groupvelth} for precise statement.

In Section 2, we give necessary definitions in order to describe the class of quasiperiodic 
operators we are going to work with. In Section 3, we formulate the main results both in the 
language of the effective Hamiltonian and of the $XY$ spin chain. We also describe 
several concrete classes of operators satisfying our assumptions. In Section 4, we summarize the properties of Aubry duality and of $L^2$ degree 0 reducible operators that are relevant for the proof of the main results. In Section 5, we prove Theorems \ref{main} and \ref{groupvelth}. In Section 6 we explain the main steps of translating the language of 
the effective Hamiltonians to the language of $XY$ spin chains and prove Corollary \ref{lieb_cor}.
\section{$L^2$-degree 0 reducible quasiperiodic operators}
The main result will be formulated for an abstract class of quasiperiodic $d$-frequency operators. In order to formulate the results, we first need to introduce this class.

Let $v\colon \T^d\to \R$ be a continuous function. We will also consider $v$ as a $\Z^d$-periodic continuous 
function on $\R^d$. A $d$-frequency one-dimensional quasiperiodic operator family is a collection of operators of the form
\beq
\label{h_def}
(H(x)\psi)_n=\psi_{n+1}+\psi_{n-1}+v(x+n\alpha)\psi_n, \quad n\in \Z^d,
\eeq
where $\alpha=(\alpha_1,\ldots,\alpha_d)\in \R^d$ is a vector of frequencies, and $n\alpha=(n_1\alpha_1,\ldots,n_d\alpha_d)$. We assume that the set $\{1,\alpha_1,\ldots,\alpha_d\}$ is linearly independent over $\Q$, in which case it is an ergodic 
operator family with respect to the dynamics $x\mapsto x+\alpha$ on $\T^d$.
The eigenvalue equation
\beq
\label{eiv_eq}
\psi_{n+1}+\psi_{n-1}+ v(x+n\alpha) \psi_n =E\psi_n
\eeq
can be written in the following form involving transfer matrices,
$$
\begin{pmatrix}\psi_n\\ \psi_{n-1}\end{pmatrix}=\l(\prod_{j=n-1}^{0}S_{v,E}(x+j\alpha)\r)\begin{pmatrix}\psi_{0}\\ \psi_{-1}\end{pmatrix},
$$
where
$$
S_{v,E}(x)=\begin{pmatrix}
E-v(x)&-1\\
1&0
\end{pmatrix},
$$
and the pair $(\alpha,S_{v,E})$ is called a Schr\"odinger cocycle understood as a map $(\alpha,S_{v,E}):\T^d\times \C^2\to
\T^d\times \C^2$ given by
$(\alpha,S_{v,E}):(x,w) \mapsto (x+\alpha,S_{v,E}(x) \cdot
w)$. 
Replacing $S_{v,E}$ with $A \in \sl2r$ gives a definition of an $\sl2r$-cocycle. 

For any Borel subset $\Delta\subset \R$, define {\it density of states measure} of the set $\Delta$ as
$$
N(\Delta):=\int\limits_{\T^d}(\mathbb E_{H(x)}(\Delta)\delta_0,\delta_0)\,dx,
$$
where $\mathbb E_{H}(\Delta)$ is the spectral projection of a self-adjoint operator $H$ in $l^2(\Z)$. 
The {\it integrated density of states} is defined as
$$
N(E):=N((-\infty,E))=N((-\infty,E]),\quad E\in \R.
$$
It is known that $N(E)=1-2\rho(E)$, where $\rho(E)$ is the fibered rotation number of the cocycle $(\alpha,S_{v,E}(x))$.

We call an operator 
family \eqref{h_def} $L^2$-{\it degree $0$ redicible} if, {\it for almost every $E$ with respect to the density of states measure}, there exists $B(\cdot,E)\in L^2(\T^d;\GL(2,\C))$ such 
that $|\det B(x,E)|=1$ and
\beq
\label{reducibility}
B(x+\alpha,E)S_{v,E}(x)B(x,E)^{-1}=A_{\star}\quad \text{for a. e.}\quad x\in \T^d,
\eeq
where 
\beq
\label{astar}
A_{\star}=\begin{pmatrix}
e^{2\pi i \rho(E)}&0\\ 0& e^{-2\pi i \rho(E)}.
\end{pmatrix}
\eeq
\section{Main results}
The following is the main result of the paper in terms of the effective quasiperiodic Hamiltonian.
\begin{theorem}
	\label{main}
	Let $H(x)$ be an $L^2$-degree $0$ reducible quasiperiodic operator family. Let also $(A\psi)_n=i(\psi_{n+1}-\psi_{n-1})$. There exists a full Lebesgue measure subset $\T_0\subset \T^d$ and a sequence $T_k\to +\infty$ as $k\to +\infty$ such that
	$$
	\frac{1}{T_k}\int_0^{T_k} e^{iH(x)t}Ae^{-i H(x)t}\,dt\sconv Q(x),\quad  \forall x\in \T_0,
	$$
	where $Q(x)$ is a bounded operator with trivial kernel for all $x\in \T_0$, and $\|Q(x)\|$ is constant on $\T_0$.
\end{theorem}
The integrated density of states $N(E)$ is a non-decreasing function of $E$. The inverse 
function $E(N)$ is defined on $[0,1]$ except maybe for a countable set of points of the form
$\{\alpha_1\Z+\ldots+\alpha_d \Z\}$ which correspond to gaps in $\sigma(H)$. We can define 
$E(N)$ arbitrarily at these points so that the resulting function is non-decreasing on $
[0,1]$ and hence differentiable almost everywhere on $[0,1]$.
\begin{theorem}
	\label{groupvelth}
	Under the assumptions of Theorem $\ref{main}$, suppose that the density of states measure of the family $H$ is absolutely continuous. Then, for almost every $x\in \T^d$,
	\beq
		\label{groupvel}
		\|Q(x)\|=\frac{1}{\pi} \esssup\limits_{[0,1]}\frac{dE}{dN}.
	\eeq
\end{theorem}
\begin{cor}
	\label{lieb_cor}
	Suppose that an isotropic $XY$ spin chain has effective Hamiltonian satisfying the 
	assumptions of Theorem $\ref{main}$. Then \eqref{liebrobbound} can only hold for all $\Lambda\subset \Z$ if $\mathfrak v\ge 2\|Q(x)\|$.
\end{cor}
\subsection{Concrete classes of operators}
Theorems \ref{main}, \ref{groupvelth} were formulated under some abstract assumptions. However, these assumptions hold for a wide class of operators.
\begin{cor}
	\label{concrete1}
	Let $H(x)$ be a one-frequency $($i.e. $d=1)$ quasiperiodic Schr\"odinger 
	operator \eqref{h_def} with $v\in C^{\omega}(\T)$, Diophantine frequency $\alpha$, and purely absolutely continuous spectrum. Then the statements of Theorems $\ref{main}$, $\ref{groupvelth}$, and Corollary $\ref{lieb_cor}$ hold.
\end{cor}
\begin{proof}
	From the results of \cite{AK1,AFK}, it follows that the operator family $H(x)$ is 
	analytically reducible for Lebesgue almost all energies for which the Lyapunov exponent vanishes. Hence, it satisfies the assumptions of mentioned theorems.
\end{proof}
\begin{remark}
	In \cite{BJ}, it is shown that if $v\in C^{\omega}(\T)$ and $\alpha$ is Diophantine, 
	then there exists $\lambda_0(v)>0$ such that the spectrum of $H(x)$ with the potential $\lambda v$ will be purely absolutely continuous for $\lambda<\lambda_0(v)$.
\end{remark}
\begin{cor}
	\label{concrete2}
	Let $H(x)$ be a multi-frequency quasiperiodic Schr\"odinger operator family with Diophantine frequency vector $\alpha$ and the potential $\lambda v$, where $v\in C^{\omega}(\T^d)$. There exists $\lambda_0(\alpha,v)>0$ such that, for $\lambda<\lambda_0(\alpha,v)$, the statements of Theorems $\ref{main}$, $\ref{groupvelth}$, and Corollary $\ref{lieb_cor}$ hold.
\end{cor}
\begin{proof}
	In \cite{Eli}, it is shown that given $v\in C^{\omega}(T)$ and a fixed Diophantine frequency vector $\alpha$, one can find $\lambda_0(\alpha,v)>0$ such 
	that $S_{\lambda v,E}$ is analytically reducible for $N$-almost every $E$ for $\lambda<\lambda_0(\alpha,v)$, from which all the statements follow.
\end{proof}




\section{Properties of $L^2$-degree 0 reducible operators}
{\it Aubry duality} is a relation between spectral properties of $H(x)$ and the dual Hamiltonian $\wt H(\theta)$ which is the following operator in $l^2(\Z^d)$:
\beq
\label{h_dual}
(\wt H(\theta)\psi)_m=\sum\limits_{m'\in \Z^d}\hat v_{m'} \psi_{m-m'}+2\cos 2\pi(\alpha\cdot m+\theta)\psi_m,
\eeq
where
$$
v(x)=\sum_{k\in\Z}\hat{v}_k e^{2\pi i k\cdot x}.
$$
Denote the corresponding direct integral spaces (for $H$ and $\wt H$ respectively) by
$$
\mathfrak{H}:=\dirintd l^2(\Z)\,dx,\quad \widetilde{\mathfrak{H}}= {\int_{\T}^{\oplus}l^2(\Z^d)}\,d\theta.
$$
Consider the unitary operator $\mathcal U\colon \mathfrak{H}\to \widetilde{\mathfrak{H}}$ defined on vector functions $\Psi=\Psi(x,n)$ as
\beq
\label{u_def}
(\mathcal {U} \Psi)(\theta,m)=\hat{\Psi}(m,x+\alpha\cdot m),
\eeq
where $\hat\Psi$ denotes the Fourier transform over $x\in \T^d\to m\in \Z^d$ combined with the inverse Fourier transform $n\in \Z\to \theta\in \T$. Let also
$$
\mathcal{H}:=\dirintd H(x)\,dx,\quad \wt{\mathcal{H}}:=\dirint \wt H(\theta)\,d\theta.
$$
Aubry duality can be formulated as the following 
equality of direct integrals.
\beq
\label{Duality}
\mathcal{U} \mathcal H \mathcal U^{-1}=\wt{\mathcal H}.
\eeq 
It is well known (see, for example, \cite{GJLS}) that the spectra of $H(x)$ and $\wt H(\theta)$ coincide for all $x,\theta$. We denote then both by $\Sigma$. 
Moreover, the IDS of the families $H$ and $\wt H$ also
coincide.We will use the following properties of the fibered rotation number and $L^2$-degree 0 reducibility, see \cite{JK} and references therein for more detail.
\begin{enumerate}
\item The rotation number is a continuous non-increasing function of $E$. It is locally constant on $\R\setminus \Sigma$, where $\Sigma=\sigma(H(x))$ (this set does not depend on $x$), and its values on $\R\setminus\Sigma$ are $\frac12 \Z$-linear combinations of $\alpha_1,\ldots,\alpha_d$. It maps $\Sigma$ onto $[0,1/2]$.
\item Suppose that \eqref{reducibility} holds for {\it some} $A_{\star}$ for $N$-almost every $E$ with $B$ continuous in $x$. Then, $N$-almost every $E$, it also holds with $A_{\star}$ given by \eqref{astar}. In other words, a continuously reducible quasiperiodic operator family is automatically degree 0 reducible.
\item For almost every $\theta\in [0,1/2]$, there exists a unique $E\in \R$ such that $\rho(E)=\theta$. 
Denote this $E$ by $E(\theta)$. For almost every $\theta\in [0,1/2]$, the cocycle $S_{v,E(\theta)}$ 
satisfies \eqref{reducibility}. By $f(x,\theta)$, denote the matrix element $(B^{-1}(x,E(\theta)))_{11}$. The function $f$ can be chosen to be $L^2$-normalized in $x$ and measurable in $\theta$. We assume that from now on.
\item Extend $f(x,\theta)$ to $\theta\in [-1/2,0]$ by $f(x,-\theta):=\overline{f(x,\theta)}$, and then extend it with period 1 to all $\theta\in \R$. Then, for almost every $\theta$, the dual Hamiltonian $\wt H(\theta)$ has purely point spectrum with eigenvalues $E(\theta-k\cdot \alpha)$, $k\in \Z^d$, and eigenvectors $u(\theta,k)_m=\hat{f}(m+k,\theta-k\alpha)$, where $\hat f$ is the Fourier transform of $f$ over $x$. We use the convention that the lower index $m\in \Z^d$ of the vector enumerates its components, $k\in \Z^d$ enumerates different eigenvectors of the same operator, and $\theta$ is the ergodic parameter enumerating the operators $\wt H(\theta)$.
\item For almost every $\theta$, the function $d(\theta)=e^{2\pi i \theta}f(x,\theta)\overline{f(x-\alpha,\theta)}-e^{-2\pi i \theta}\overline{f(x,\theta)}f(x-\alpha,\theta)$ is well defined (i.e. does not depend on $x$ almost surely) and non-zero. For these $\theta$, one can choose $B(x)^{-1}$ of the form
$$
\frac{1}{d(\theta)^{1/2}}\begin{pmatrix}
f(x,\theta) &\overline{f(x,\theta)}\\ e^{-2\pi i \theta} f(x-\alpha,\theta)& 
e^{2\pi i \theta} \overline{f(x-\alpha,\theta)}
\end{pmatrix}.
$$
\end{enumerate}
In the sequel, we will denote the matrix $B(x)$ obtained for $E=E(\theta)$ by $B(x,\theta)$.

\section{Proofs of Theorems \ref{main} and \ref{groupvelth}}
We first need a few auxiliary results.
Suppose that $\mathfrak H=\dirintd l^2(\Z)\,dx$. We will consider bounded decomposable operators of the form $\mathcal H=\dirintd H(x)\,dx$, where $H(\cdot)$ is an a. e. uniformly bounded measurable family of operators.
\begin{prop}
	\label{strong_conv}
	Suppose that $\mathcal H_n=\dirintd H_n(x)\,dx$ is a sequence of bounded decomposable operators. Then
	\begin{enumerate}
		\item If $\|H_n(x)\|\le C$ for a. e. $x\in \T^d$ and all $n\in \N$, and $H_n(x)\sconv H(x)$ for a. e. $x\in \T^d$, then $\mathcal H_n\sconv \dirintd H(x)\,dx$.
		\item If $\mathcal H_n\sconv \mathcal H$, 
		where $\mathcal H$ is a bounded operator 
		on $\dirintd l^2(\Z)\,dx$, then $\mathcal H$ is a bounded decomposable operator, and there exists a subset $\T_0\subset \T^d$ of full Lebesgue measure and a subsequence $\{n_k\}$ such that $H_{n_k}(x)\sconv H(x)$ for $x\in \T_0$.
		\item
		$\mathcal H$ has trivial kernel if and only if $H(x)$ has trivial kernel for a. e. $x\in \T^d$.
	\end{enumerate}
\end{prop}
\begin{proof}
	The first claim follows from the dominated convergence theorem applied to  the integral 
	$$
	\int_T \|H_n(x)f(x)-H(x)f(x)\|\,dx,
	$$
	where $f\colon \T^d\to l^2(\Z)$ is an element of $\dirintd l^2(\Z)\,dx$.
	To prove the second claim, denote by $e_k(\cdot)\in \dirint l^2(\Z)\,dx$ the constant vector function $x\mapsto \delta_k$. We have $\H_{n} e_k(\cdot)\sconv \H e_k(\cdot)$ in $\dirint l^2(\Z)\,dx$. Hence, $\int_{\T^d}\|H_n(x)e_k-H(x)e_k\|^2\,dx\to 0$. Since $L^2$-convergence implies convergence in measure, we get that there exists a subsequence such that $H_{n_l}(x)e_k\to H(x)e_k$ for almost every $x$. Applying Cantor 
	diagonal procedure, we can ensure that there exists a set of $x$ of full measure and a subsequence that converges on all basis vectors $e_k$. Since this sequence is also uniformly bounded, by 
	Banach--Steinhaus theorem, there will be strong operator convergence on that subsequence. Third claim is well known.\,\qedhere
\end{proof}

\begin{lemma}
	\label{pointspectrum}
	Let $A$ be a bounded operator and $H$ be a bounded self-adjoint operator with purely point spectrum. Let $\{\lambda_l\}$ be the distinct eigenvalues of $H$, and $P_l={\mathbb E}_H\{\lambda_l\}$ be the projection onto the corresponding eigenspace. Then
	$$
	\slim_{T\to+\infty}\frac{1}{T}\int_0^T e^{iHt}A e^{-iHt}\,dt=\sum\limits_l P_l A P_l.
	$$
\end{lemma}
The right hand side can be considered as the ``diagonal part'' of $A$ with respect to the eigenspaces of $J$.
\begin{proof}
	The left hand side is uniformly bounded in $T$. Due to Banach--Steinhaus theorem, it is sufficient to check the convergence on the eigenvectors of $H$. Let $\psi_{l,j}$ be the eigenvectors of $H$, $J\psi_{l,j}=\lambda_l\psi_{l,j}$, where $j$ enumerates 
	different eigenvectors from the same eigenspace. We have
	$$
	\frac{1}{T}\int_0^t e^{iHt}A e^{-iHt}\,dt\,\psi_{l,j}=\sum_{l',k} \frac{1}{T}\int_0^T  e^{i(\lambda_{l'}-\lambda_l)t}\,dt(A\psi_{l,j},\psi_{l',k})\psi_{l',k}=
	$$
	$$
	=P_l A P_l \psi_{l,j}+\sum_{l'\neq l}\sum_k \frac{e^{i(\lambda_{l'}-\lambda_l)T}-1}{Ti(\lambda_{l'}-\lambda_l)}(A\psi_{l,j},\psi_{l',k})\psi_{l',k}=
	$$
	$$
	=P_l A P_l \psi_{l,j}+\sum_{{l'}\neq l}e^{i(\lambda_{l'}-\lambda_l)T/2}\frac{\sin((\lambda_{l'}-\lambda_l)T/2)}{(\lambda_{l'}-\lambda_l)T/2}P_{l'} A\psi_{l,j}.
	$$
	To show that the last sum converges to $0$, note that the terms are mutually orthogonal, and $\sum_{l'\neq l}\|P_{l'} A \psi_{l,j}\|^2\le \|A\psi_{l,j}\|^2$. Hence, given $\er>0$, there exists $N$ such that the sum over $l'> N$ is norm bounded by $\er$ for 
	all $T>0$. The sum over $l'<N$ has finite number of terms each of which goes to 0 as 
	$T\to +\infty$.
\end{proof}
\begin{lemma}
	\label{dual_convergence}
	Under the assumptions of Theorem $\ref{main}$, let $\wt H(\theta)$ be the dual operator family. Let also $(\wt A(\theta)\psi)_n=2\sin (2\pi (n\alpha+\theta))\psi_n$. Then, for almost every $\theta\in \T$, there exists a strong limit
	$$
	\wt Q(\theta)=\slim\limits_{T\to \infty}\frac{1}{T}\int_0^T e^{i\wt H(\theta)t}\wt A(\theta)e^{-i\wt H(\theta)t}\,dt,
	$$
	$\wt Q(\theta)$ has trivial kernel for almost every $\theta\in \T$.
\end{lemma}
\begin{proof}
	Suppose that $\theta$ is chosen to fulfill Properties 3--5 from the previous subsection. Then the operator $\wt H(\theta)$ has purely point spectrum, and the convergence is established by Lemma \ref{pointspectrum}. We only need to show that $\ker\wt Q(\theta)=\{0\}$. The operator $\wt Q(\theta)$ is diagonal in the basis $u(\theta,k)$ of 
	eigenvectors of $\wt H(\theta)$. Let us compute its diagonal entries, that is,
	\beq
		\label{wtq}
		\wt Q(\theta)u(\theta,k)=\l\{\sum\limits_{m\in \Z^d}2\sin(2\pi(m\cdot \alpha+\theta))|u(\theta,k)_m|^2\r\} u(\theta,k),\quad k\in \Z^d,
	\eeq
	where $u(\theta,k)_m=\hat{f}(m+k,\theta-k\alpha)$ are normalized eigenvectors of $\wt H(\theta)$. We need to show that none of the diagonal entries (i.e. sums in curly brackets) are zero. Let us first take $k=0$, and denote $u(\theta)_m=u(\theta,0)_m$.
	Then
	$$
		\sum_m 2\sin(2\pi(m\cdot \alpha+\theta))|u(\theta)_m|^2=
		2\im \sum_m e^{2\pi i (m\cdot \alpha+\theta)}\hat{f}(m,\theta)\overline{\hat{f}(m,\theta)}
	$$
	$$
	= 2\im \int_{\T^d} e^{2\pi i \theta}f(x+\alpha,\theta)\overline{f(x,\theta)}\,dx
	$$
	$$
	=\int_{\T^d} \l\{e^{2\pi i \theta}f(x+\alpha,\theta)\overline{f(x,\theta)}-e^{-2\pi i \theta}\overline{f(x+\alpha,\theta)}{f(x,\theta)}\r\}\,dx=d(\theta)\neq 0
	$$
	for almost every $\theta$ by Property 5. The case $k\neq 0$ is obtained by replacing $\theta$ with 	$\theta+k\cdot \alpha$. The set of $\theta$ such that the last quantity is non-zero for all $k\in \Z$ has full measure as a countable intersection of 
	full measure sets.
\end{proof}
\begin{remark}
	From the proof, it is easy to see that $\|\wt Q(\theta)\|=\|\wt Q(\theta+k\cdot \alpha)\|$, and that it is a measurable function of $\theta$. Hence, $\|\wt Q(\theta)\|$ is almost surely constant in $\theta$.
\end{remark}

\noindent {\bf Proof of Theorem \ref{main}. } Let $\mathcal U$ be the duality operator \eqref{u_def}. We have
$$
\mathcal U \l(\dirintd e^{i H(x) t}Ae^{-i H(x)t}\,dx\r)\mathcal U^{-1}=
\dirint  e^{i \wt H(\theta) t}\wt A(\theta) e^{-i \wt H(\theta)t}\,d\theta.
$$
Since the operators $\wt H(\theta)$ have purely point spectra for almost all $\theta$, Lemma \ref{dual_convergence} and Proposition \ref{strong_conv} imply that the Cesaro averages of the right hand side converge to a 
bounded operator with non-zero kernel 
which we will denote by $\wt {\mathcal Q}$. Hence,
$$
\frac{1}{T}\int_0^T \l\{\dirintd e^{i H(x) t}Ae^{-i H(x)t}\,dx\r\}\,dt\sconv \mathcal Q=\mathcal U^{-1}\wt {\mathcal Q}\mathcal U.
$$
By Proposition \ref{strong_conv}, $\mathcal Q$ is decomposable, and so $\mathcal Q=\dirint Q(x)\,dx$, where $Q(x)$ has trivial kernel for almost every $x$. Existence of a 
subsequence follows from Proposition \ref{strong_conv}, and the fact that $\|Q(x)\|$ is almost surely constant follows from the 
fact that $Q(x+\alpha)$ is unitary equivalent to $Q(x)$.\,\qed
\vskip 1mm
\noindent {\bf Proof of Theorem \ref{groupvelth}. } The argument is fairly standard, 
and was used (with \eqref{dnde} obtained), for example, in \cite{AFK} for the case of analytically reducible cocycles. For almost every $\theta$, we can consider
\beq
\label{b_def}
B(x,\theta)^{-1}=\frac{1}{|d(\theta)|^{1/2}}\begin{pmatrix}
	f(x,\theta) &\overline{f(x,\theta)}\\ e^{-2\pi i \theta} f(x-\alpha,\theta)& 
	e^{2\pi i \theta} \overline{f(x-\alpha,\theta)}
\end{pmatrix},
\eeq
where $d(\theta)=e^{2\pi i \theta}f(x,\theta)\overline{f(x-\alpha,\theta)}-e^{-2\pi i \theta}\overline{f(x,\theta)}f(x-\alpha,\theta)$ does not depend on $x$. We have
$$
B(x+\alpha,\theta)S_{v,E}(x)B(x,\theta)^{-1}=\begin{pmatrix}
e^{2\pi i \theta}&0\\ 0& e^{-2\pi i \theta}
\end{pmatrix},
$$
and $\det B(x,\theta)=\pm i$. Take $J:=\mp\frac{1}{\sqrt{2}}\begin{pmatrix}
1&1 \\ -i&i
\end{pmatrix}$. Define the new matrix $\wt B(x,\theta):=JB(x,\theta)\in \sl2r$. We have
$$
\wt B(x+\alpha,\theta)S_{v,E}(x)\wt B(x,\theta)^{-1}=\begin{pmatrix}
\cos\theta&-\sin\theta\\ \sin\theta& \cos\theta\end{pmatrix}=:R_{\theta}.
$$
From the proof of Lemma \ref{dual_convergence} and \eqref{b_def}, it 
follows that, for almost all $\theta$,
\beq
\label{qsup}
\|\wt Q(\theta)\|=\sup_{k\in \Z}|d(\theta+k\alpha)|=\sup_{k\in \Z}\frac{4}{\int_{\T^d}\|B(x,\theta+k\alpha)^{-1}\|_{\mathrm{HS}}^2\,dx}=\sup_{k\in \Z}\frac{4}{\int_{\T^d}\|\wt B(x,\theta+k\alpha)\|_{\mathrm{HS}}^2\,dx},
\eeq
where $\mathrm{HS}$ denotes the Hilbert--Schmidt norm; note that $B$ and $B^{-1}$ have the 
same norms.

We now need to recall some results from Kotani theory, see, for example, \cite{D_rev}.
The formula
\beq
\label{dnde}
\frac{dN}{dE}=\frac{1}{2\pi} \int_{\T^d}\frac{1}{\im m(E,x)}\,dx
\eeq
is valid for almost every $E$ for which $\gamma(E)=0$. Here $m$ is the $m$-function, or a measurable invariant 
section of the hyperbolic action of $S_{v,E}(x)$ with the properties
$$
\im m(E,x)>0,\quad m(E,x+\alpha)=S_{v,E}(x)\cdot m(E,x),
$$
where ``$\cdot$'' denotes the hyperbolic action of an $\sl2r$-matrix on the upper half plane, that is,
$$
\begin{pmatrix}
a&b\\c&d
\end{pmatrix}\cdot z=\frac{az+b}{cz+d}.
$$
Note that $L^2$-reducibility of $S_{v,E}$ implies that $\gamma(E)=0$. Since we assume that $N$ is absolutely continuous, Kotani's formula is valid for almost every $E$ with respect to $N$. If $C(x+\alpha)S_{v,E}(x)C(x)^{-1}\in \mathrm{SO}(2,\mathbb R)$ with {\it some} $C\in \sl2r$, then one can check that
\beq
\label{hyp_dos}
\|C(x)\|^2_{\mathrm{HS}}=\frac{1}{\im C(x)^{-1}\cdot i}+\frac{1}{\im C(x+\alpha)^{-1}\cdot i}.
\eeq

Kotani's theory also implies that, for Lebesgue almost every $E$ with $\gamma(E)=0$ (which, in our notation, would also be for almost every $\theta$), there exists $C(\cdot,\theta)\in L^2(\T^d,\sl2r)$ such that
$$
C(x+\alpha,\theta)S_{v,E}(x)C(x,\theta)^{-1}\in \mathrm{SO}(2,\mathbb R),\quad C(x,\theta)^{-1}\cdot i=m(x,E).
$$
We claim that, even though our matrix $\wt B$ may be different from $C$, we also have $\wt B(x,\theta)^{-1}\cdot i=m(x,E)$. Indeed, take 
$u(x):= \wt B(x,\theta)C(x,\theta)^{-1}\cdot i$.
We have $R_{\theta}\cdot u(x)=u(x+\alpha)$. The set of $x\in \T^d$ such that $u(x)=i$ has 
either zero or full measure (since $R_{\theta}$ preserves $i$). Assume that 
it has zero measure and take $w(x)=\frac{u(x)-i}{u(x)+i}$. A simple computation shows that
$w(x)$ is a measurable unitary function satisfying $e^{2\pi i \theta}w(x)=w(x+\alpha)$. 
This is only possible if $\theta\in \alpha_1\Z+\ldots \alpha_d\Z+\Z$. If we exclude these $\theta$, we can assume that $u(x)=i$ for almost every $x$, and hence $\wt B(x,\theta)^{-1}\cdot i=m(x,E)$. This implies that
$$
\frac{dN}{dE}=\frac{1}{4\pi}\int_{\T^d} \|\wt B(x,\theta)\|^2_{\mathrm{HS}}\,dx.
$$
From \eqref{qsup}, we obtain that
$$
\|\wt Q(\theta)\|=\sup_{k\in \Z}\frac{1}{\pi}\l.\l(\frac{dN}{dE}\r)^{-1}\r|_{E=E(\theta+k\alpha)}=\esssup\limits_{[0,1]} \frac{1}{\pi}\frac{dE}{dN}
$$
for almost every $\theta$.
\section{Proof of Corollary \ref{lieb_cor}}
Corollary \ref{lieb_cor} follows from the following result.
\begin{theorem}
	\label{lieblower}
	Suppose that an isotropic $XY$ spin chain satisfies the assumptions of Theorem $\ref{liebrob}$, and, for some real sequence $T_k\to +\infty$,
	\beq
\label{qlimit}
	\frac{1}{T_k}\int_0^{T_k} e^{iH_{\eff}t}Ae^{-i H_{\eff}t}\,dt\sconv Q.
	\eeq
	Then, \eqref{liebrobbound} can only hold for all $\Lambda\subset \Z$ if $\mathfrak v\ge 2\|Q\|$.
\end{theorem}
This result was essentially proved in \cite{D_per} with two minor differences. 
The first one is that we only assume that the limit exists on a subsequence. 
This difference is really minor and it can be easily traced that the proof remains the same. 
The other difference is that we only consider isotropic spin chains, in which case the 
effective Hamiltonian decouples. This simplifies some computations, and we choose to 
include some proofs in order to keep the text more self-contained.

\begin{prop}\cite{D_per}
	\label{transport}
	Under the assumptions of Theorem $\ref{lieblower}$, for any $\er>0$ there exists $k(\er)\in \N$ and constants $C(\er)$, $L(\er)$ such that for any $k\ge K(\er)$ one can find $l,r\in \Z$ with $|l|\le K(\er)$ and
	$$
	(\|Q\|-\er)T_k\le |r|\le (\|Q\|+\er)T_k
	$$
	such that
	$$
	|(\delta_r,e^{-i T_k H_{\eff}}\delta_l)|^2\ge \frac{C}{T_k}.
	$$
\end{prop}
\begin{proof}
	Let $v=\|Q\|$. Without loss of generality, one can assume that $\chi_{[v-\er/2,v]}(Q)\neq 0$. There exists $l\in \Z$, $l\le K(\er)$, such that $\chi_{[v-\er/2,v]}(Q) \delta_l\neq 0$. Let us now relate the operator $Q$ with transport properties. Denote by $X$ the position operator in $l^2(\Z)$,
	$$
	(Xu)_n=n u_n,
	$$
	defined on the natural domain. The operator $Q$ is related with the Heisenberg evolution 
	of $X$. We have $\dom X(T)=\dom X$, and
	$$
	X(T)u=X u+\int_0^T A(t) u\,dt=Xu+\int_0^T e^{iH t}Ae^{-iH t}\,dt,\quad u\in \dom X.
	$$
	where
	$$
	A=\overline{i[X,H]},\quad (A\psi)_n=-i\psi_{n-1}+i\psi_{n+1}.
	$$
	We have $\frac{1}{T_k} X(T_k)u\to Qu$ for any $u\in \dom X$ as $k\to \infty$. This implies that the 
	sequence $\frac{1}{T_k}X(T_k)$ converges to $Q$ {\it in the strong resolvent sense}, see Theorem VIII.25 from \cite{RS1}. Let $\varphi$ be a continuous non-negative function 
	equal to 1 on $[v-\er/2,v]$ and vanishing outside $[v-\er,v+\er]$. 
	Due to Theorem VIII.20 from \cite{RS1}, we have
	$$
	\varphi\l(\frac{1}{T_k} X(T_k)\r) u\to \varphi\l(Q\r)u,\quad \forall u\in \dom X,
	$$
	and so 
	$$
	\|\chi_{[T_k(v-\er),T_k (v+\er)]}(X(T_k))\delta_l\|\ge \l\|\varphi\l(X(T_k)/T_k\r)\delta_l\r\|\ge C
	$$ 
	for sufficiently large $k$.
	The indicator function in the left hand side is simply the sum of projections onto $\delta_r$ with $r\in [T_k(v-\er),T_k (v+\er)]$. Hence,
	$$
	\sum\limits_{r\in [T_k(v-\er),T_k (v+\er)]} |(\delta_r,e^{-iT_k H}\delta_l)|^2\ge C,
	$$
	and so, for some $r\in [T_k(v-\er),T_k (v+\er)]$, we have
	$$
		|(\delta_r,e^{-iT_k H}\delta_l)|^2\ge \frac{C}{2 \er T_k+1}\ge \frac{C_1}{T_k}.\,\qedhere
	$$
	\end{proof}
	Let us now consider the $XY$ spin chain on $\Lambda=[m,n]\cap \Z$, and construct the 
	following observables
	$$
	a_j^*:=\frac12 (\sigma_j^x+i\sigma_j^y),\quad a_j:=\frac12 (\sigma_j^x-i\sigma_j^y),
	$$
	$$
	c_m:=a_m,\quad c_{m+j}:=\sigma_m^z\sigma_{m+1}^z\ldots\sigma _{m+j-1}^z a_{m+j},\quad 1\le j\le n-m.
	$$
	Let also $H_{\eff}^{\Lambda}:=\l. H_{\eff}\r|_{[m,n]}$ be the restriction of the original operator onto $[m,n]\cap \Z$.
	\begin{prop}
		\label{observables}
		Let $m\le l\le r\le n$. Then
		$$
		\|[c_l(t),a_r^*]\|\ge |(e^{-2itH_{\eff}^{\Lambda}}\delta_l,\delta_r)|.
		$$
	\end{prop}
	\begin{proof}
		It was shown in \cite{HSS} that, if $C^{\Lambda}=(c_m,\ldots,c_n)^T$, then
		$$
			C^{\Lambda}(t)=e^{-2 i t H_{\eff} }C^{\Lambda}.
		$$
		Consider the following special state
		$$
		u_{\Lambda}=\bigotimes_{j=m}^n \begin{pmatrix}
		1\\0
		\end{pmatrix}\in \mathfrak G_{\Lambda}.
		$$ 
		A simple application of commutation relations between $c_l$ and $a_r^*$ shows that
		$$
			\|[c_l(t),a_r^*]\|\ge |[c_l(t),a_r^*]u_{\Lambda}|=|(e^{-2itH_{\eff}^{\Lambda}}\delta_l,\delta_r)|,
		$$
		which completes the proof.
	\end{proof}
	{\noindent \bf Proof of Theorem 6.1.} Suppose that Lieb--Robinson bound holds with the 
	velocity $\mathfrak v$. Fix some $\er>0$ and obtain $l,r$ from Proposition \ref{transport}. 
	Due to Proposition \ref{observables} with $t=T_k/2$ and because of Lieb--Robinson bound, we have
	$$
		|(e^{-iT_k H_{\eff}^{\Lambda}}\delta_l,\delta_r)|\le\|[c_l(t),a_r^*]\|\le Ce^{-\eta(|r-l|-\mathfrak v T_k/2)}.
	$$
	This inequality must hold for all $\Lambda\subset \Z$, and hence (after taking strong limit) we have the following for $H_{\eff}$ if $k\ge K(\er)$:
	$$
	\frac{C}{\sqrt{T_k}}\le |(e^{-iT_k H_{\eff}}\delta_l,\delta_r)|\le C_1 e^{-\eta(|r-l|-\mathfrak v T_k/2)}\le C_1 e^{-\eta((\|Q\|-\er)T_k-L(\er)-1-\mathfrak{v}T_k/2)}.
	$$
	This inequality must hold for arbitrarily large $k$, which is only possible if $\mathfrak v\ge 2(\|Q\|-\er)$. Since $\er$ is arbitrary, this completes the proof.\,\qed
	
\begin{remark}
\label{ballexp_rem}
The (non-averaged) lower transport exponent is defined as
$$
\beta^-_{\psi}(p)=\liminf\limits_{T\to +\infty}\frac{\log(X(T)\psi,\psi)}{p\log t}.
$$
One can show that, under the assumptions of Theorem \label{lieblower}, we have
$$
\liminf\limits_{k\to\infty}\frac{1}{T_k^p}\l(|X(T_k)|\psi,\psi\r)\ge (|Q|^p\psi,\psi)>0,
$$
for any $\psi\neq 0$ with finite support. If the limit \eqref{qlimit} existed as 
$T\to\infty$ (not on a subsequence), one would be able to show that $\beta_{\psi}^-(p)=1$ for any $p>0$. In the case of quasiperiodic operators as in Theorem \ref{main}, one can show that
\beq
\label{ballexp}
\liminf\limits_{T\to\infty}\frac{1}{T}\int\limits_{\T^d}\l(|X(T,x)|^p\psi,\psi\r)\,dx\ge \int\limits_{\T^d}(|Q(x)|^p\psi,\psi)>0,
\eeq
where $X(T,x)=e^{iT H(x)}Xe^{-iT H(x)}$. Since $x$ is the ergodic parameter, 
\eqref{ballexp} can be understood as presence of ballistic transport in expectation.
\end{remark}
	\section{Acknowledgements} I would like to thank Svetlana Jitomirskaya for drawing my attention to
	the problem. I would also like to thank her, as well as Vojkan Jaksic, Milivoje Lukic and Gunter St\"olz, for valuable discussions. I am grateful to 
	Isaac Newton Institute for Mathematical Sciences, Cambridge, for support
	and hospitality during the programme Periodic and Ergodic Spectral
	Problems where part of this paper was completed. The research was supported by 
	AMS--Simons Travel Grant, 2014--2016, and partially by NSF DMS--1401204.

\end{document}